\documentclass[letterpaper,12pt]{article} 

\usepackage[utf8]{inputenc}
\usepackage{geometry}
\usepackage{setspace}

\usepackage{fullpage}
\usepackage{geometry}

\usepackage{fullpage}
\usepackage[T1]{fontenc}
\usepackage{enumitem}
\usepackage{ae,aecompl}
\usepackage[]{graphicx}
\usepackage{amssymb}
\usepackage{amsmath}
\usepackage[english]{babel}
\usepackage[round,authoryear]{natbib}
\usepackage{epsf}
\usepackage{xcolor}
\usepackage{array}
\usepackage{rotating}
\usepackage{ushort}
\usepackage{adjustbox}
\usepackage{ntheorem}
\usepackage{tikz}
\definecolor{webgreen}{rgb}{0,0.4,0}
\definecolor{webbrown}{rgb}{0.6,0,0}
\definecolor{purple}{rgb}{0.5,0,0.25}
\definecolor{darkblue}{rgb}{0,0,0.7}
\definecolor{darkred}{rgb}{0.7,0,0}
\DeclareMathOperator*{\esssup}{ess\,sup}
\usepackage[pdfborder=false]{hyperref}
\hypersetup{colorlinks,citecolor=darkred,filecolor=black,linkcolor=darkblue,urlcolor=webgreen,pdfpagemode=None,pdfstartview=FitH}
\newcommand{\ignore}[1]{}
\newtheorem{lemma}{{\textsc{Lemma}}}

\newtheorem{cor}{{\textsc{Corollary} }}
\newtheorem{thm}{{\textsc{Theorem} }}
\newtheorem{defn}{{\textsc{Definition} }}

\newtheorem*{theorem-non}{Theorem}
\newenvironment{proof}{\noindent { \textbf{Proof} \/}:\enspace}
{\hfill $\blacksquare{}$ \vspace{12pt}}

\makeatletter
\renewcommand\section{\@startsection {section}{1}{\z@}%
                                   {-3.5ex \@plus -1ex \@minus -.2ex}%
                                   {2.3ex \@plus.2ex}%
                                   {\centering\large\scshape}}
\renewcommand\subsection{\@startsection {subsection}{1}{\z@}%
                                   {-3.5ex \@plus -1ex \@minus -.2ex}%
                                   {2.3ex \@plus.2ex}%
                                   {\raggedright\large\scshape}}
\renewcommand\subsubsection{\@startsection {subsubsection}{1}{\z@}%
                                   {-3.5ex \@plus -1ex \@minus -.2ex}%
                                   {2.3ex \@plus.2ex}%
                                   {\raggedright\itshape}}

\usepackage{nicematrix}

\makeatother

\title{Optimal Sequential Assignment with Capacity Constrained Verification }

\author{
Vilok Taori\thanks{Indian Statistical Institute, Delhi, India.}
}

\begin{document}
\maketitle
\date{}

 \begin{center}
        \large\textbf{Abstract}
    \end{center}
   A principal seeks to allocate $k$ identical objects among $n$ sequentially arriving, impatient agents. Each agent privately observes her valuation, and the principal's payoff from allocating an object depends on the recipient's valuation. The principal can perfectly verify the valuation of at most $m$ agents, where $m < k$, and commits to a mechanism before any agent arrives. We characterize the class of optimal mechanisms. Our main result establishes that, in states where the number of remaining objects exceeds the number of available verification checks, the optimal mechanism may involve probabilistic verification.

  \footnote{ I thank Siddharth Chatterjee, Albin Erlanson, Debasis Mishra, Alessandro Pavan and Arunava Sen for valuable discussions and suggestions.}

   \section{Introduction}
   We study an allocation problem in which a principal assigns $k$ identical objects to a group of impatient agents who arrive sequentially. Each agent desires an object, and the principal’s payoff from allocation depends on the agent’s privately known valuation. Monetary transfers are not available to elicit private information, though valuations are based on verifiable evidence. The principal possesses a verification technology that can perfectly check this evidence, but its capacity is limited to fewer than the number of objects to be allocated. The optimal mechanism that maximises the principal’s total expected payoff must account for the frictions arising from sequential arrival and the capacity constraint on verification. Allocating an object today by verifying an agent may preclude allocating it in the future to a higher-valued agent, while also consuming a verification check that could otherwise be used later.
   
As an illustrative example, consider a firm hiring for a set of identical positions. Applications arrive sequentially, and the firm’s payoff from assigning a position to an applicant depends on privately known but verifiable attributes such as education, work experience, and skills. The firm’s verification resources are limited, with fewer checks available than positions to fill. The firm’s objective is to design an optimal allocation policy.

Assuming agents’ valuations are independently distributed, past decisions can be classified into states, where each state specifies the set of agents yet to arrive, the number of objects remaining, and the number of verification checks available. The mechanism design problem is to determine, given past decisions, whether and how to verify the arriving agent based on her report, and how to use the verification outcome (if obtained) together with the report to allocate an object, while accounting for its implications for future allocations.

The optimal mechanism must balance two opposing forces. In states where the number of objects exceeds the available verification checks, the principal faces a fundamental trade-off: verifying the arriving agent raises the quality of the current allocation by screening out low types, but consumes a check that could have been used to screen future agents. Allocating without verification conserves screening capacity for the future, but admits low-quality agents today. Our main result shows that these two forces exactly cancel in such states — the gain from selective allocation today is precisely offset by the loss in future screening capacity — so that the optimal expected payoff is the same under the two mechanisms: (i) allocate the object to the arriving agent without verification, thereby conserving a check, or (ii) verify the agent if her reported type exceeds a threshold and allocate if the report is truthful. Consequently, any randomisation between these two mechanisms — resulting in probabilistic verification — also yields the optimal payoff. This balance breaks down when the number of objects equals the number of available checks: here, every object can potentially be filled selectively, the opportunity cost of consuming a check is higher, and threshold verification strictly dominates.

   \section{Literature Review}
The literature on sequential assignment began with the paradigmatic analysis by \cite{derman1972sequential}. They consider a model of complete information, where several heterogeneous objects are allocated over time to sequentially arriving agents. Our model differs from theirs as we assume that the arriving agents' types are private information. Also, we restrict attention to the sequential assignment of identical objects.  \cite{gershkov2010efficient} also study the sequential assignment problem where arriving agents' types are private information, but they allow for monetary transfers. They show that the dynamic efficient policy can be implemented by a dynamic analogue of the VCG (Vickrey-Clarke-Groves) mechanism. Our model does not allow for monetary transfers but assumes that private information is based on verifiable evidence. 

Our paper relates to the recent work on allocation problems where private information is based on verifiable evidence. \cite{ben2014optimal} study an optimal allocation problem where a principal desires to allocate a single indivisible good among a group of agents where agent types can be learnt at a cost. Subsequent work has studied costly verification without transfers in collective decisions (\cite{erlanson2020costly}), delegation (\cite{halac2020commitment}), or general mechanism design problems (\cite{ben2019mechanisms}). \cite{erlanson2024optimal} study a static allocation problem with capacity-constrained verification. They consider a setting where the principal has to allocate a set of identical objects among a group of agents who do not arrive sequentially. 

In dynamic contexts, \cite{popov2016stochastic} studies a model of repeated risk sharing with costly verification. \cite{DLib} study a dynamic principal-agent problem with costly verification. The closest related work to ours is \cite{epitropou2019optimal} which studies the optimal allocation of a single indivisible good among sequentially arriving agents when the principal can learn agent types at a cost. Our model differs in two important respects. First, we consider the allocation of multiple identical objects rather than a single good, which introduces a richer set of trade-offs: the principal must decide not only whether to verify the current agent but also how to allocate the remaining verification budget across future arrivals. Second, rather than a cost-based verification technology, we impose a hard capacity constraint on the number of verifications available, which means the principal cannot simply purchase additional information but must instead allocate a fixed and exhaustible screening resource across agents. It is this combination - multiple objects and a constrained verification technology that gives rise to the payoff equivalence between threshold verification and unverified allocation, a phenomenon that has no counterpart in the single-object setting of \cite{epitropou2019optimal}.

The rest of the paper is organised as follows: we present the model in section \ref{section:3}. Section \ref{sec 4} describes the optimal mechanism and contains the main result. In section \ref{sec:6}, we derive the properties of the optimal mechanism. We conclude in section \ref{sec:7}

\section{Model}\label{section:3}
There are $k $ identical positions to be allocated. Each agent $i$ is characterised by a type $t$ that is private information. The agents arrive in a sequential order: first agent $1$ appears, then agent $2$, \ldots, and then agent $n$ where $n > k$. Each agent can only be allocated the position upon arrival. After a position is allocated, it cannot be reallocated in the future.

An agent with type $t$ who is allocated the position receives payoff $u_i(t)$ and a payoff of $0$ if not allocated the position. If a position is allocated to agent with type $t$, then the payoff to the principal is $t$. Agents' types are independently and identically distributed random variables $T_i$ with a common atomless c.d.f $F$, full support on $[0, 1]$ and a mean denoted by $\mu$.\footnote{The results of the paper follow for any non-negative bounded support} Let $T$  be  the product of individual random variables, i.e. $T = \prod_{i \in I}T_{i}$.  The principal cannot use transfers but can perfectly verify the type of $m$ agents, where $m < k$.

The principal's problem is to determine the optimal mechanism that maximises the expected payoff from allocating $k$ positions to $n$ agents who arrive sequentially, given a verification technology that can perfectly learn the types of $m$ agents. The principal commits to the mechanism. Histories of past decsions correspond to the events where some positions are yet to be allocated. As agents' types are independently distributed, all possible histories of past decisions can be classified in terms of different states, where each state is represented by the arriving agent $i$, the number of positions yet to be allocated, and the number of verification checks remaining. For example, the initial history can be represented by the state ($1,k,m$). Let $S_i$ be the collection of all states involving agent $i$ for $i \in \{1,\ldots ,n\}$.

For the optimal mechanism, without loss of generality, we can restrict attention to the following incentive compatible direct mechanism.\footnote{This can be shown using arguments similar in \cite{ben2014optimal}.} Fix a state, and consider an agent who arrives and reports a type. The mechanism then (possibly randomly) decides whether to check this report. If checked, the agent is allocated the position if and only if she reported truthfully. If not checked, the mechanism allocates the position randomly. 

Hence, for any state $s \in S_i$, the mechanism specifies for each report, the probability with which the agent is verified and the probability with which the agent is allocated the position if she is not verified. Let $p(t|s)$ be the probability that the agent with report $t$ is verified and allocated the position if reported truthfully, and $q(t|s)$ be the probability that she is not verified and allocated the position. Since all agents have non-negative types, for any state $s$ in which the number of remaining agents equals the number of available positions, every remaining agent is allocated the position without verification.  For any state $s$ where the number of available checks is $0$, without loss of optimality, the remaining positions are allocated to each arriving agent irrespective of her report. Henceforth, the mechanism will be denoted by $(p,q)$ and its restriction to state $s \in S_i$ will be denoted by $(p(.|s), q(.|s))$. A mechanism also describes the transition probabilities between states. Given a state $s = (i, k',m')$ and a report  $t$, the states $(i+1, k'-1, m'-1)$, $(i+1, k'-1, m')$ and $(i + 1, k', m')$ are reached with probabilities $p(t|s),\;\;q(t|s)$ and $(1-p(t|s)-q(t|s)$ respectively. For every state $s$ and report $t$, the mechanism must satisfy
\[
    p(t|s)\geq 0, \qquad q(t|s)\geq 0, \qquad
    p(t|s)+q(t|s)\leq 1.
\]The initial state and mechanism $(p,q)$ make each state $s \in S_i$ for each $i \in \{2,\ldots,n\}$ a random variable.

The principal's objective is
\begin{equation} \label{object}
    E_T \left[ \sum_{i =1}^{n} \sum_{s \in S_i} (p(t|s) + q(t|s))t\right]
\end{equation}
The incentive compatibility constraints for agent $i$ are
\begin{equation*}
  p(t|s) + q(t|s) \geq   q(t^{'}|s) \;\; \text{for every} \;\; t, t^{'} \in [0,1] \;\; and\; s\in S_i
\end{equation*}
 
\section{Optimal Mechanism}\label{sec 4}

As agents are short-lived, the incentive compatibility conditions impose constraints on the mechanism for a given state and not across states. However, the principal's choice of the mechanism influences the transition probabilities across states. Hence, the principal's objective \ref{object} is a dynamic optimisation problem with a finite number of states.

By the principle of optimality of dynamic programming, the optimal mechanism can be constructed using backward induction. We first determine the optimal incentive-compatible mechanism restricted to all possible final states, then extend the optimal mechanism to all states with a transition to a final state, and continue in this manner until the optimal incentive-compatible mechanism for the initial state is determined.

The principal's objective function can now be restated to determine an optimal incentive-compatible mechanism restricted to each state, given that she will follow the optimal mechanism for all possible states that may be realised later. Let $V(s)$ be the optimal expected value to the principal at state $s$ when the principal commits to the optimal mechanism.

Fix a state $s = (i,k',m')$, and let $c(s)$ denote the difference in optimal expected future value between the case in which agent $i$ is not verified and not allocated the position and the case in which agent $i$ is verified and allocated the position:
\begin{equation} \label{c1}
    c(s) = V(i+1,k',m') - V(i+1,k'-1,m'-1).
\end{equation}
We now define the class of sequential threshold mechanisms.

\begin{defn}
A mechanism $(p,q)$ is a \emph{sequential-threshold mechanism} if its restriction $(p(.|s),q(.|s))$ to every state $s=(i,k',m')$, $m' \geq 1$, satisfies the following conditions.\footnote{Recall that, for any state, if the number of remaining agents equals the number of available positions, every remaining agent is allocated the position without verification.}

\[
p(t \mid s) =
\begin{cases}
0, & t < c(s), \\
1 - \beta(s), & t > c(s),
\end{cases}
\qquad
q(t \mid s) = \beta(s), \text{for all} \, \;t \in [0,1],
\]
where $\beta(s) \in [0,1]$
\end{defn}

Since $F$ is atomless, mechanisms that differ only on an $F$-null set of reports
induce the same expected payoff and the same transition probabilities. Throughout,
we identify such mechanisms. Similarly, the values of $p(\cdot|s)$ and
$q(\cdot|s)$ at cutoff types are payoff-irrelevant and may be chosen arbitrarily.
We are now ready to state our first theorem.

\begin{thm}\label{Thm1}
Every optimal mechanism belongs to the class of sequential threshold mechanisms.
\end{thm}

 To build intuition for the proof, note that the incentive constraints 
require the total allocation probability at any given type to weakly exceed the unverified 
allocation probability at any other type. This infinite family of pairwise 
constraints turns out to be equivalent to the existence of a \textbf{single scalar} 
$\beta \in [0,1]$, the supremum of the unverified allocation probability 
across all types, such that the unverified allocation probability at every type is 
at most $\beta$ and the total allocation probability at every type is at least 
$\beta$.

With $\beta$ fixed, the problem \textbf{separates across types}: since both the 
objective and the constraints involve the choices at each type independently of the 
choices at any other type, the principal can optimise pointwise at each type 
separately. This pointwise optimisation has a clean solution: it is optimal to set 
the unverified allocation probability equal to $\beta$ at every type, and to verify 
if and only if the reported type exceeds the threshold $c(s) = V(i+1,k',m') - 
V(i+1,k'-1,m'-1)$. The threshold emerges directly from the sign of $t - c(s)$: 
above the threshold, verification raises the principal's payoff and so is worthwhile; 
below the threshold, it does not.
The resulting mechanism takes a sequential threshold form at every state.

\begin{proof}
    
Fix a non-terminal state $s = (i, k', m')$ with $1 \leq m' \leq k'$.
 The principal's payoff from type $t$ at state $s$ is
\begin{eqnarray}
    \Pi(t)
&=&
p(t|s)\bigl(t + V(i+1,k'-1,m'-1)\bigr) \nonumber \\
&+& q(t|s)\bigl(t + V(i+1,k'-1,m')\bigr) \nonumber \\
&+& (1-p(t|s)-q(t|s))\,V(i+1,k',m') \nonumber
\end{eqnarray}
Defining
\[
\alpha_1(t) := t + V(i+1,k'-1,m'-1) - V(i+1,k',m') = t - c(s),
\]
\[
\alpha_2(t) := t + V(i+1,k'-1,m') - V(i+1,k',m'),
\]
the payoff simplifies to
\[
\Pi(t) = V(i+1,k',m') + \alpha_1(t)\,p(t|s) + \alpha_2(t)\,q(t|s).
\]
Observe, $\alpha_1(t) = t - c(s)$,  is 
positive above the threshold $c(s)$ and negative below it. 

The incentive constraint
\[
p(t \mid s) + q(t \mid s) \geq q(t' \mid s)
\quad \text{for every }  t, t'
\]
This is equivalent to the existence of a 
scalar $\beta \in [0,1]$ such that
\[
q(t|s) \leq \beta
\qquad \text{and} \qquad
p(t|s) + q(t|s) \geq \beta
\quad \text{for every }  t.
\]
Let $\beta = \esssup_{t^{'}} q(t'|s)$,  the highest unverified allocation probability across all types. Then,
 $q(t|s) \leq \beta$ holds by definition of the supremum, and $p(t|s) + q(t|s) \geq \beta$ follows from the incentive constraint, since the total allocation probability at any $t$ must weakly exceed the unverified allocation probability at any other type, including the type achieving the supremum.

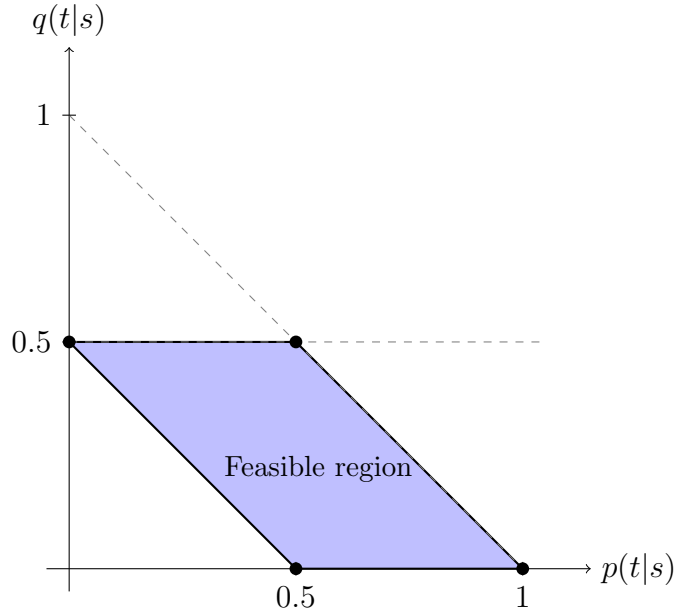
\begin{figure}[htbp]
\centering
\begin{tikzpicture}[scale=6]

  \def\beta{0.5} 

  \fill[blue!25] (\beta,0) -- (1,0) -- (1-\beta,\beta) -- (0,\beta) -- cycle;

  \draw[->] (-0.05,0) -- (1.15,0) node[right] {$p(t|s)$};
  \draw[->] (0,-0.05) -- (0,1.15) node[above] {$q(t|s)$};

  \draw[thick] (\beta,0) -- (1,0);          
  \draw[thick] (1,0) -- (1-\beta,\beta);    
  \draw[thick] (1-\beta,\beta) -- (0,\beta);
  \draw[thick] (0,\beta) -- (\beta,0);      

  \draw[dashed, gray] (0,1) -- (1,0);       
  \draw[dashed, gray] (0,\beta) -- (1.05,\beta); 

  \draw (1,0.015) -- (1,-0.015) node[below] {$1$};
  \draw (0.015,1) -- (-0.015,1) node[left] {$1$};
  \draw (\beta,0.015) -- (\beta,-0.015) node[below] {$\beta$};
  \draw (0.015,\beta) -- (-0.015,\beta) node[left] {$\beta$};

  \fill (\beta,0) circle (0.4pt);
  \fill (1,0) circle (0.4pt);
  \fill (1-\beta,\beta) circle (0.4pt);
  \fill (0,\beta) circle (0.4pt);

  \node at (0.55,0.22) {\small Feasible region};

\end{tikzpicture}
\caption{Feasible region defined by $p(t|s)\ge 0$, $q(t|s)\ge 0$, $p(t|s)+q(t|s)\le 1$, $q(t|s)\le \beta$, and $p(t|s)+q(t|s)\ge \beta$ for $\beta =0.5$.}
\label{fig:feasible_region}
\end{figure}
The principal's objective function can now be restated as 
\[
\max_{\beta \in [0,1]} \max_{\substack{q(t|s)\leq \beta,\\ p(t|s) + q(t|s) \geq \beta}}\int_0^{1} \bigl[V(i+1,k',m') + \alpha_1(t)\,p(t|s) + \alpha_2(t)\,q(t|s)\bigr]\,dF(t)
\]
Since both the objective and the constraints, for fixed $\beta$, involve the choices 
at each type $t$ independently of the choices at any other type, the problem 
separates across types. For each $t$ independently, we choose 
$(p(t|s),q(t|s))$ subject to $q(t|s) \leq \beta$ and $p(t|s) + q(t|s) \geq \beta$ 
to maximise $V(i+1,k',m') + \alpha_1(t)\,p(t|s) + \alpha_2(t)\,q(t|s)$.

Let $\sigma := p(t|s) + q(t|s) \in [\beta, 1]$ denote the total allocation 
probability. For fixed $\sigma$, the objective is
\[
V(i+1,k',m') + \alpha_1(t)\,\sigma 
+ \bigl(V(i+1,k'-1,m') - V(i+1,k'-1,m'-1)\bigr)\,q(t|s).
\]
Lemma~\ref{eqn:extra} in the Appendix establishes that the optimal value function 
 is increasing in the number of verification checks, with equality attained precisely when the number of remaining agents equals the number of available objects; hence $V(i+1,k'-1,m') - V(i+1,k'-1,m'-1) \geq 0$. The 
objective is increasing in $q(t|s)$, so it is optimal to set $q(t|s) = \beta$. 
This gives $p(t|s) = \sigma - \beta$, and the remaining problem is
\[
\max_{\sigma \in [\beta,1]}\;
V(i+1,k',m') + \alpha_1(t)\,\sigma 
+ \bigl(V(i+1,k'-1,m') - V(i+1,k'-1,m'-1)\bigr)\,\beta.
\]
The coefficient of $\sigma$ is $\alpha_1(t) = t - c(s)$. Hence:
\begin{enumerate}[label=\roman*)]
    \item  If $t < c(s)$: $\alpha_1(t) < 0$, so set $\sigma = \beta$, giving 
    $p(t|s) = 0$ and $q(t|s) = \beta$.
    \item  If $t > c(s)$: $\alpha_1(t) > 0$, so set $\sigma = 1$, giving 
    $p(t|s) = 1-\beta$ and $q(t|s) = \beta$.
    \item  At $t = c(s)$: every $\sigma \in [\beta,1]$ is optimal; since $F$ is 
    atomless, this point is payoff-irrelevant.
\end{enumerate}

Hence, proved.
\end{proof}

The optimal mechanism must specify $\beta(s)$ at every state $s$.
The value induced at state $s$ by a choice of $\beta \in [0,1]$ is
\begin{align*}
\Psi_s(\beta)
&=
\int_0^{c(s)}
\left[
\beta\bigl(t + V(i+1,k'-1,m')\bigr) + (1-\beta)V(i+1,k',m')
\right]dF(t) \\
&\quad +
\int_{c(s)}^{1}
\left[
t + (1-\beta)V(i+1,k'-1,m'-1) + \beta V(i+1,k'-1,m')
\right]dF(t).
\end{align*}
Hence, the optimal mechanism is characterised by pointwise optimising at every state $s$:
\[
\max_{\beta \in [0,1]} \Psi_s(\beta).
\]
\noindent
We now define a sub-class of sequential threshold mechanisms:

\begin{defn}
A mechanism $(p,q)$ is a \emph{aligned-sequential threshold mechanism} if its restriction $(p(.|s),q(.|s))$ to every state $s=(i,k',m')$, $m' \geq 1$, satisfies the following conditions.

\[
p(t \mid s) =
\begin{cases}
0, & t < c(s), \\
1 - \beta(s), & t > c(s),
\end{cases}
\qquad
q(t \mid s) = \beta(s), \text{for all}\; t \in [0,1]
\]
where $\beta(s) = \begin{cases}
0, & k'=m', \\
\rho(s)\in[0,1], & k'>m',
\end{cases}$
\end{defn}
\noindent
We are now ready to state our main theorem:
\begin{thm}\label{Thm1}
Every aligned-sequential threshold mechanism is an optimal mechanism. 
\end{thm}

We provide an intuitive justification of the payoff equivalence between the different aligned-sequential threshold mechanisms. Consider the state $s = (i, k', m')$ where the number of positions to allocate strictly exceeds the number of available verification checks and the following deterministic sequentially aligned-threshold mechanism restricted to state $s$,
\begin{enumerate}[label=\roman*.]
    \item \textbf{Allocate without verifying ($\beta(s) =1$):} The principal allocates the position to the agent without using the verification check, conserving it for future agents.
    \item \textbf{Threshold verification ($\beta(s) =0$):} The principal verifies the agent if and only if her reported type exceeds the cut-off $c(s)$, allocates the position if the report is confirmed truthful, and does not allocate otherwise.
\end{enumerate}

Theorem \ref{Thm1} establishes that whenever the number of positions exceeds the available verification checks, the expected payoff from threshold verification equals the expected payoff from allocating without verification. The payoff equivalence between the two mechanisms reflects a precise balancing of two opposing forces. In the \textbf{threshold verification} mechanism restricted to state $s$, the principal is selective: she screens out low types below the cut-off $c(s)$ and only allocates to agents whose types are confirmed to exceed $c(s)$. This selectivity raises the average quality of the current allocation. However, it comes at a cost: using the verification check now leaves fewer checks for future agents, reducing the principal's ability to screen future arrivals. The principal is therefore trading quality today against screening capacity tomorrow. In the \textbf{allocate without verifying} mechanism restricted to state $s$ , the principal conserves the verification check for future agents, preserving screening capacity.
However, she must allocate the current position without any information, inevitably admitting low-quality agents. The saved check improves future screening, but the current allocation is of lower average quality. The reason these two forces exactly cancel is that the value of a verification check is precisely equal to the value of being able to allocate one additional position selectively. 

This balance is specific to states where $k' > m'$, i.e.\ where positions outnumber checks. In such states, the principal will eventually have to allocate some positions without verification, regardless, so the option value of saving a check is limited. When $k' = m'$, an unverified
allocation leaves $k'-1$ positions and  $k'$ checks, so one check is redundant; a verified allocation instead leaves
 $k'-1$ of each, without reducing future selective capacity. This absence of a continuation cost is why threshold verification dominates.
\noindent
We now proceed to prove the Theorem.

\begin{proof}
Fix a non-terminal state $s = (i, k', m')$ with $1 \leq m' \leq k'$, we will prove the following:
\begin{enumerate}[label=\roman*)]
    \item If $k' = m'$, then $\Psi_s(\beta)$ is uniquely maximised at $\beta = 0$. 
    Hence \textbf{threshold verification} is uniquely optimal.
    \item If $k' > m'$, then $\Psi_s(\beta)$ is constant on $[0,1]$. Hence 
    \textbf{threshold verification}, \textbf{allocation without verification}, and 
    every probabilistic mixture between the two are payoff equivalent and optimal.
    \end{enumerate}
Since $\Psi_s(\beta)$ is affine in $\beta$, we can write
\[
\Psi_s(\beta) - \Psi_s(0) = \beta D_s,
\]
where $D_s$ is the derivative of $\Psi_s$ with respect to $\beta$, computed as:
\begin{align*}
D_s 
&= \int_0^{c(s)} (t + V(i+1,k'-1,m') - V(i+1,k',m'))\,dF(t) \\
&+ \int_{c(s)}^{1} (V(i+1,k'-1,m') - V(i+1,k'-1,m'-1))\,dF(t)
\end{align*}
Using $c(s) = V(i+1,k',m') - V(i+1,k'-1,m'-1)$, this simplifies to:
\[
D_s = \int_0^{c(s)} t\,dF(t) + (1-F(c(s)))c(s) - (V(i+1,k',m') - V(i+1,k'-1,m'))
\]
The sign of $D_s$ determines the optimal $\beta$: if $D_s < 0$ then $\beta = 0$ is 
uniquely optimal; if $D_s = 0$ then every $\beta \in [0,1]$ is optimal; and if 
$D_s > 0$ then $\beta = 1$ is uniquely optimal.

\medskip
\noindent\textbf{Part (i): $k' = m'$.}
After an unverified allocation, the successor state is $(i+1, k'-1, k')$, which has 
$k'-1$ positions but $k'$ checks. Since checks cannot exceed positions, the extra 
check is redundant:
\[
V(i+1,k'-1,k') = V(i+1,k'-1,k'-1),
\]
and so, since $m' = k'$:
\[
V(i+1,k'-1,m') = V(i+1,k'-1,m'-1).
\]
Therefore $V(i+1,k',m') - V(i+1,k'-1,m') = c(s)$, and:
\[
D_s
=
\int_0^{c(s)} t\,dF(t) + (1-F(c(s)))c(s) - c(s)
=
\int_0^{c(s)} (t - c(s))\,dF(t) < 0,
\]
where the strict inequality holds because $c(s) > 0$ and $F$ has full support, so 
the integrand $t - c(s)$ is strictly negative on a set of positive measure. Hence 
$\Psi_s(\beta)$ is strictly decreasing in $\beta$, and $\beta = 0$ --- threshold 
verification --- is uniquely optimal.

\noindent\textbf{Part (ii): $k' > m'$.}
We prove by backward induction on the number of remaining agents the following two statements for every state $(i, k', m')$ with $k' > m'$:
\begin{enumerate}[label=(\roman*)]
    \item $D_s = 0$.
    \item $V(i, k', m') = \mu + V(i+1, k'-1, m')$.
\end{enumerate}
Statement (ii) says that whenever positions outnumber checks, the value of the 
optimal mechanism equals the value of allocating the current position without 
verification and continuing optimally.

We first prove the base case, when the number of remaining agents equals the number of remaining positions. Every 
remaining agent must receive a position, so verification is payoff-irrelevant and
\[
V(i, k', m') = k'\mu = \mu + V(i+1, k'-1, m').
\]
Statement (ii) holds. Since the allocation decision is forced regardless of $\beta$, 
the value does not depend on $\beta$, so $D_s = 0$ trivially and Statement (i) 
holds.

We now consider the inductive step.
Fix a state $s = (i, k', m')$ with $k' > m'$ and assume both (i) and (ii) hold for all states that can transit from state $s$. We establish (i) and (ii) at state $s$.
By the induction hypothesis, for the states $(i+1,k',m')$ and $(i+1,k'-1,m'-1)$, we have 
\[
V(i+1,k',m') = \mu + V(i+2,k'-1,m'),
\qquad
V(i+1,k'-1,m'-1) = \mu + V(i+2,k'-2,m'-1).
\]
Subtracting $V(i+1,k'-1,m'-1)$ from $V(i+1, k', m')$, we obtain
\begin{equation}\label{eqp1}
c(s) = V(i+1,k',m') - V(i+1,k'-1,m'-1) = V(i+2,k'-1,m') - V(i+2,k'-2,m'-1).
\end{equation}
Hence, the threshold at state $s$ equals the threshold at the transition state 
$(i+1,k'-1,m')$.
For this state, either $k'-1 = m'$ (in which case Part (i) 
applies) or $k'-1 > m'$ (in which case the induction hypothesis applies). In both 
cases, threshold verification with threshold $c(s)$ is optimal, giving:
\begin{equation}\label{eqnp2}
V(i+1,k'-1,m') 
= F(c(s))\,V(i+2,k'-1,m') + \int_{c(s)}^{1} \bigl[t + V(i+2,k'-2,m'-1)\bigr]\,dF(t).
\end{equation}
Now,
\begin{align}
&V(i+1,k',m') - V(i+1,k'-1,m') \nonumber\\
&= \mu + V(i+2,k'-1,m') 
- F(c(s))\,V(i+2,k'-1,m') 
- \int_{c(s)}^{1} \bigl[t + V(i+2,k'-2,m'-1)\bigr]\,dF(t)\nonumber \\
&= \mu + c(s) - F(c(s))\,c(s) - \int_{c(s)}^{1} t\,dF(t) \nonumber\\
&= \int_0^{c(s)} t\,dF(t) + (1-F(c(s)))\,c(s) \label{eqp3}
\end{align}
where the first equality follows from the induction hypothesis and equation \ref{eqnp2}, the second equality follows from equation \ref{eqp1} and the last equality uses $\mu = \int_0^{c(s)} t\,dF(t) + \int_{c(s)}^{\bar{t}} t\,dF(t)$.
Substituting the expression \label{eqp3} for $V(i+1,k',m') - V(i+1,k'-1,m')$ into $D_s$, we obtain
\[
D_s 
= \int_0^{c(s)} t\,dF(t) + (1-F(c(s)))\,c(s) 
- \int_0^{c(s)} t\,dF(t) - (1-F(c(s)))\,c(s) = 0.
\]
Hence, statement (i) is established.
Since $D_s = 0$, every $\beta \in [0,1]$ is optimal. In particular, $\beta = 1$,
allocation without verification is optimal, giving:
\[
\Psi_s(1) = \int_0^{1} \bigl(t + V(i+1,k'-1,m')\bigr)\,dF(t) = \mu + V(i+1,k'-1,m').
\]
Hence $V(i,k',m') = \mu + V(i+1,k'-1,m')$, establishing Statement (ii) at the 
current state. This completes the induction step and the proof of the Theorem.
\end{proof}

Using the two theorems above, we can characterize the class of optimal mechanisms.
\begin{cor}
A mechanism is optimal if and only if it is an aligned sequential-threshold
mechanism.
\end{cor}
\section{Properties of the Threshold Verification Mechanism}\label{sec:6}
We now provide a recursive formula to determine the cutoff and study how the cutoff of the threshold-verification representative varies
across non-trivial states, that is, states in which the number of remaining agents
strictly exceeds the number of remaining positions. Write $r$ for the number of
remaining agents. Thus the value function and cutoff are denoted by
$V(r,k',m')$ and $c(r,k',m')$.

For non-trivial states with $r>k'\geq m'\geq 1$, the threshold-verification
representative satisfies
\[
V(r,k',m')
=
\int_{c(r,k',m')}^1
\left[t+V(r-1,k'-1,m'-1)\right]\,dF(t)
+
F(c(r,k',m'))V(r-1,k',m'),
\]
and the cutoff is
\[
c(r,k',m')
=
V(r-1,k',m')-V(r-1,k'-1,m'-1).
\]

We use the following boundary conventions. If the number of remaining agents equals
the number of remaining positions, then every remaining agent must be allocated, so
\[
V(k',k',m')=k'\mu.
\]
For the purpose of the recursive cutoff formula below, set
\[
c(k',k',m')=0
\qquad\text{for } k'\geq m'\geq 1.
\]
If there are no verification checks left, set
\[
V(r,k',0)=k'\mu
\qquad\text{and}\qquad
c(r,k',0)=1
\]
for all $r>k'\geq 0$. The convention $c(r,k',0)=1$ is only a boundary convention:
with no checks left, there is no actual verification cutoff.

\begin{thm}
 For each non-trivial
state with $r>k'\geq m'\geq 1$,  the following comparative statics hold.

\begin{enumerate}
    \item The cutoff increases with the number of remaining agents:
    \[
    c(r+1,k',m')>c(r,k',m').
    \]
    \item Whenever $r>k'+1$, the cutoff decreases with the number of remaining
    positions:
    \[
    c(r,k'+1,m')<c(r,k',m').
    \]
    \item Whenever $k'>m'$, the cutoff decreases with the number of remaining
    verification checks:
    \[
    c(r,k',m'+1)<c(r,k',m').
    \]
\end{enumerate}
\end{thm}

\begin{proof}
We first derive the recursive formula for cutoffs that will be used to show the comparative statics. Fix a non-terminal state $(r,k',m')$ with $r>k'\geq m'\geq 1$. Let
$a:=c(r-1,k',m')$ and $b:=c(r-1,k'-1,m'-1).$
 Define
\[
H(x):=\int_x^1 t\,dF(t)-(1-F(x))x.
\]
Using the value recursion and the cutoff equation, we have
\[
V(r,k',m')-V(r-1,k',m')=H(c(r,k',m'))
\]
for all non-trivial states with $m'\geq 1$. The same identity also holds at the
no-check boundary if we use the convention $c(r,k',0)=1$, because then
\[
V(r,k',0)-V(r-1,k',0)=0=H(1).
\]
Therefore,
\[
\begin{aligned}
c(r,k',m')
&=V(r-1,k',m')-V(r-1,k'-1,m'-1)\\
&=\left[V(r-2,k',m')-V(r-2,k'-1,m'-1)\right]\\
&\quad+\left[V(r-1,k',m')-V(r-2,k',m')\right]\\
&\quad-\left[V(r-1,k'-1,m'-1)-V(r-2,k'-1,m'-1)\right]\\
&=a+H(a)-H(b).
\end{aligned}
\]
We define
\[
T(a,b) =
a+H(a)-H(b)
=
F(a)a+\int_a^b t\,dF(t)+(1-F(b))b
\]
Thus, the recursive formula for the cut-offs is given by:
\[
c(r,k',m')
=
T\left(c(r-1,k',m'),c(r-1,k'-1,m'-1)\right).
\]

We record two elementary properties of $T$. First, $T$ is strictly increasing
in each argument. Indeed, if $0\leq a_1<a_2\leq b$, then integration by parts gives
\[
T(a_2,b)-T(a_1,b)=\int_{a_1}^{a_2} F(x)\,dx>0,
\]
where the strict inequality follows from full support. Similarly, if
$a\leq b_1<b_2\leq 1$, then
\[
T(a,b_2)-T(a,b_1)=\int_{b_1}^{b_2} (1-F(x))\,dx>0.
\]
Second, if $a<b$, then $T(a,b)$ lies strictly between $a$ and $b$:
\[
a<T(a,b)<b.
\]
To see this, note that
\[
T(a,b)-a
=
\int_a^b (t-a)\,dF(t)+(b-a)(1-F(b))>0,
\]
and
\[
b-T(a,b)
=
(b-a)F(a)+\int_a^b (b-t)\,dF(t)>0.
\]

We now prove the comparative statics by induction on $r$. The boundary values
$c(k',k',m')=0$ and $c(r,k',0)=1$ initialise the recursion. The recursive formula
above, together with the fact that $T(a,b)$ lies strictly between $a$ and $b$ whenever
$a<b$, implies in particular that all non-trivial cutoffs lie strictly between $0$ and
$1$.

Assume that the comparative statics with respect to positions and checks have been
established for states with $r-1$ remaining agents. We first prove the comparative
static with respect to positions at states with $r$ remaining agents. Fix
$r>k'+1$ and $k'\geq m'\geq 1$. By the recursive formula,
\[
c(r,k',m')
=
T\left(c(r-1,k',m'),c(r-1,k'-1,m'-1)\right),
\]
whereas
\[
c(r,k'+1,m')
=
T\left(c(r-1,k'+1,m'),c(r-1,k',m'-1)\right).
\]
By the induction hypothesis, or by the boundary convention when the comparison
state is terminal,
\[
c(r-1,k',m')>c(r-1,k'+1,m').
\]
Moreover,
\[
c(r-1,k'-1,m'-1)\geq c(r-1,k',m'-1),
\]
with strict inequality whenever $m'>1$; when $m'=1$, both terms are equal to the
no-check boundary value $1$. Since $T$ is strictly increasing in each argument, and
the first inequality is strict, it follows that
\[
c(r,k',m')>c(r,k'+1,m').
\]
This proves the comparative static with respect to the number of positions.

Next fix $r>k'\geq m'+1\geq 2$. By the recursive formula,
\[
c(r,k',m')
=
T\left(c(r-1,k',m'),c(r-1,k'-1,m'-1)\right),
\]
and
\[
c(r,k',m'+1)
=
T\left(c(r-1,k',m'+1),c(r-1,k'-1,m')\right).
\]
By the induction hypothesis, or by the terminal boundary convention if
$r-1=k'$,
\[
c(r-1,k',m')\geq c(r-1,k',m'+1).
\]
Also, by the induction hypothesis applied to the state with $k'-1$ positions,
\[
c(r-1,k'-1,m'-1)>c(r-1,k'-1,m').
\]
The second inequality is strict. Since $T$ is increasing in both arguments and
strictly increasing in the second one, we obtain
\[
c(r,k',m')>c(r,k',m'+1).
\]
This proves the comparative static with respect to the number of verification checks.

It remains to prove the comparative static with respect to the number of remaining
agents. Fix a non-trivial state with $r>k'\geq m'\geq 1$. By the recursive formula,
\[
c(r+1,k',m')
=
T\left(c(r,k',m'),c(r,k'-1,m'-1)\right).
\]
We claim that
\[
c(r,k'-1,m'-1)>c(r,k',m').
\]
If $m'=1$, this follows from the boundary convention
$c(r,k'-1,0)=1$ and the fact that every non-trivial cutoff lies strictly below $1$.
If $m'>1$, then the two comparative statics already established at the state with
$r$ remaining agents imply
\[
c(r,k'-1,m'-1)>c(r,k',m'-1)>c(r,k',m').
\]
Thus, in all cases,
\[
c(r,k'-1,m'-1)>c(r,k',m').
\]
Since $T(a,b)$ lies strictly between $a$ and $b$ whenever $a<b$, we have
\[
c(r+1,k',m')
=
T\left(c(r,k',m'),c(r,k'-1,m'-1)\right)
>
c(r,k',m').
\]
This proves the comparative static with respect to the number of remaining agents
and completes the induction.
\end{proof}

\section{Conclusion} \label{sec:7}
The main finding of this paper is that when the number of available objects exceeds the remaining checks, a sequential threshold mechanism—sometimes granting items probabilistically without verification to preserve screening capacity—emerges as optimal. This mechanism balances the trade-off between selective allocation in the present and conserving verification resources for the future.
In future work, we aim to examine the strategic interactions that arise when the principal does not commit to disclosing whether a check is available. Another important direction is to study the optimal sequential allocation of heterogeneous objects, as heterogeneity introduces new trade-offs and complexities into the allocation problem.
\section{Appendix}\label{sec:appendix}

\begin{lemma}\label{eqn:extra}
$V(i, k', m') \geq V(i,k',m'-1)$
where $i \geq 1$, $1\leq k' \leq k$ and $1 \leq m' \leq m$.
\end{lemma}

    \begin{proof}
Fix state $s = (i, k', m')$ where the number of remaining agents is strictly greater than the available positions. We first claim that $V(i, k', 1) > V(i, k', 0)$ 
where $i \leq n$ and $k' \leq k$. The optimal mechanism at state $(i, k', 0)$ is to 
allocate the position randomly to any of the $k'$ remaining agents. Hence, 
$V(i, k', 0) = k'\mu$.

Consider the following mechanism (that may not be optimal) at state $(i, k', 1)$. 
The arriving agent is asked to report her type. If the reported type $t_i$ is above 
$\mu$, then the agent is verified and allocated the position if the report is 
truthful. Otherwise, the agent is not verified and not allocated the position. Based 
on the realisation of $t_i$, either state $(i+1, k', 1)$ or state 
$(i+1, k'-1, 0)$ is reached. Irrespective of the state reached, the policy randomly 
allocates the remaining positions without verifying, even if a verification check is 
available. The expected value if the principal follows this policy is:
\[
    F(\mu)\,k'\mu + (1 - F(\mu))\bigl((k'-1)\mu + \mathbb{E}(t_i \mid t_i > \mu)\bigr)
\]
which is strictly greater than $k'\mu$. Hence, we have $V(i, k', 1) > V(i, k', 0)$.

Now let $\varphi$ be the optimal mechanism for the principal at state $(i, k', m'-1)$, 
such that $\varphi(s_{i''})$ denotes the mechanism at state $s_{i''} = 
(i'', k'', m'')$ where $(i'', k'', m'')$ is a state that can be reached when the 
principal employs the optimal policy starting from state $(i, k', m'-1)$. Now, 
consider a policy $\psi$ at state $(i, k', m')$ (that may not be optimal) such that 
$\psi(s_{i''}) = \varphi(i'', k'', m''-1)$ for all states $(i'', k'', m'')$ where 
$m'' \geq 2$, and $\psi(i'', k'', 1)$ is the optimal mechanism at state 
$(i'', k'', 1)$. As shown above, $V(i'', k'', 1) > V(i'', k'', 0)$ for all 
$i'' \geq i$ and $k'' \leq k'$. This implies that the expected value at state 
$(i, k', m')$ from the policy $\psi$ is strictly greater than $V(i, k', m'-1)$. 
Hence,
\[
    V(i, k', m') > V(i, k', m'-1)
\]
If the number of remaining agents equals the number of remaining positions, then
every remaining agent must receive a position. Hence verification is redundant, and
\[
    V(i,k',m')=V(i,k',m'-1)=k'\mu.
\]
\end{proof}

\bibliography{bibligraphy.bib}
\bibliographystyle{apalike}
\end{document}